\documentclass{article}

\usepackage{arxiv}

\usepackage[ruled,linesnumbered]{algorithm2e}
\usepackage{amsthm}
\usepackage{amsmath}
\usepackage{amsfonts,amssymb}
\usepackage{booktabs}
\usepackage{CJK}
\usepackage{multirow}
\usepackage{tabularx}
\usepackage{tikz,hyperref}
\usepackage{diagbox}
\usepackage{scalerel}
\usepackage{stackengine,wasysym}

\usepackage[utf8]{inputenc} % allow utf-8 input
\usepackage[T1]{fontenc}    % use 8-bit T1 fonts
\usepackage{hyperref}       % hyperlinks
\usepackage{url}            % simple URL typesetting
\usepackage{booktabs}       % professional-quality tables
\usepackage{amsfonts}       % blackboard math symbols
\usepackage{nicefrac}       % compact symbols for 1/2, etc.
\usepackage{microtype}      % microtypography
\usepackage{cleveref}       % smart cross-referencing
\usepackage{lipsum}         % Can be removed after putting your text content
\usepackage{graphicx}
\usepackage{doi}

\SetKwRepeat{Do}{do}{while}

\theoremstyle{plain}
\newtheorem{theorem}{Theorem}

\newtheorem{lemma}{Lemma}
\newtheorem{proposition}{Proposition}

\theoremstyle{definition}
\newtheorem{definition}{Definition}

\theoremstyle{definition}

\newtheorem{fact}{Fact}

\title{Distributed Quantum Discrete Logarithm Algorithm}

\newif\ifuniqueAffiliation
\ifuniqueAffiliation % Standard variant of author block
\author{ \href{https://orcid.org/0000-0000-0000-0000}{\includegraphics[scale=0.06]{orcid.pdf}\hspace{1mm}David S.~Hippocampus}\thanks{Use footnote for providing further
		information about author (webpage, alternative
		address)---\emph{not} for acknowledging funding agencies.} \\
	Department of Computer Science\\
	Cranberry-Lemon University\\
	Pittsburgh, PA 15213 \\
	\texttt{hippo@cs.cranberry-lemon.edu} \\
	\And
	\href{https://orcid.org/0000-0000-0000-0000}{\includegraphics[scale=0.06]{orcid.pdf}\hspace{1mm}Elias D.~Striatum} \\
	Department of Electrical Engineering\\
	Mount-Sheikh University\\
	Santa Narimana, Levand \\
	\texttt{stariate@ee.mount-sheikh.edu} \\
}
\else
\usepackage{authblk}

\author[1,2,4]{Renjie Xu}
\author[1,2,4]{Daowen Qiu \thanks{Corresponding author: \texttt{issqdw@mail.sysu.edu.cn}}}
\author[1]{Ligang Xiao}
\author[3,4]{Le Luo}
\author[5]{Xu Zhou}

\affil[1]{Institute of Quantum Computing and Computer Theory, School of Computer Science and Engineering, Sun Yat-sen University, Guangzhou, 510006, China.}
\affil[2]{The Guangdong Key Laboratory of Information Security Technology, Sun Yat-sen University, Guangzhou, 510006, China.}
\affil[3]{School of Physics and Astronomy, Sun Yat-sen University, Zhuhai, 519082, China}
\affil[4]{Shenzhen Research Institute of Sun Yat-Sen University, Shenzhen, 518057, China.}
\affil[5]{Yangtze Delta Industrial Innovation Center of Quantum Science and Technology, Suzhou 215000, China}

\fi

\renewcommand{\shorttitle}{Distributed Quantum Discrete Logarithm Algorithm}

\hypersetup{
pdftitle={Distributed Quantum Discrete Logarithm Algorithm},
pdfsubject={},
pdfauthor={},
pdfkeywords={Quantum Computing, Discrete Logarithm Problem, Distributed Algorithm},
}

\begin{document}
\maketitle

\begin{abstract}
    Solving the discrete logarithm problem (DLP) with quantum computers is a fundamental task with important implications. Beyond Shor’s algorithm, many researchers have proposed alternative solutions in recent years. However, due to current hardware limitations, the scale of DLP instances that can be addressed by quantum computers remains insufficient. To overcome this limitation, we propose a distributed quantum discrete logarithm algorithm that reduces the required quantum register size for solving DLPs. Specifically, we design a distributed quantum algorithm to determine whether the solution is contained in a given set. Based on this procedure, our method solves DLPs by identifying the intersection of sets containing the solution. Compared with Shor’s original algorithm, our approach reduces the register size and can improve the success probability, while requiring no quantum communication.
\end{abstract}

\section{Introduction}
Discrete Logarithm Problems (DLPs) \cite{DLP_1}, a classical puzzle in cryptography, have been shown to admit efficient quantum algorithms. However, practical implementations of quantum algorithms for DLPs are constrained by the hardware limitations of current quantum devices, in terms of both problem scale and achievable success probability. Therefore, since Shor proposed the first quantum discrete logarithm algorithm \cite{DLP_first}, many improved algorithms \cite{Exact_DLP_1, DLP_oracle_Magic_Box} have been developed to enhance the success probability or reduce the size of the control register.

To enhance the success probability of finding a discrete logarithm, two lines of work have been explored. Some researchers have proposed exact algorithms \cite{Exact_DLP_1} by introducing amplitude amplification \cite{AA}. With amplitude amplification, these works address the approximation error in the inverse Quantum Fourier Transform (QFT), thereby yielding an exact algorithm for the overall problem. The other line of work does not achieve exact results. In this line, an oracle is designed based on the half-bit approximation \cite{DLP_oracle}. Specifically, the success probability can be increased by improving the approximation to the ``half-bit'' of the DLP. Unfortunately, the half-bit problem can be very hard to approximate, because the relevant eigenvector resides in a superposition state.

To reduce the size of the control registers, most existing algorithms proceed in two ways, e.g., leveraging classical computation \cite{DLP_Classical} and reusing the input register \cite{DLP_oracle_Magic_Box, HSP}. With the support of classical computation, intermediate information can be stored classically, thus reducing the number of required qubits. Furthermore, including the oracle-based algorithms, researchers have noticed that information in registers can be reused to reduce redundant qubits. Specifically, the method of Kaliski et al. \cite{DLP_oracle_Magic_Box} for the half-bit approximation reduces the size of one register to a single qubit, and the method of Mosca et al. \cite{HSP} reduces one register by exploiting the structure of modular operations. However, both solutions rely on the assumption that one knows which eigenvector is present in the third register, while identifying this eigenvector could be no easier than solving the original problem.

Meanwhile, Distributed Quantum Algorithms (DQA) have been used to establish quantum advantage or reduce problem size without requiring difficult quantum-memory movement and addressing across locations \cite{DQA_possible}. Moreover, to reduce the complexity of Grover's algorithm, Qiu et al. \cite{DQA_Qiu} proposed a distributed version of Grover's algorithm. Tan et al. \cite{DQA_Tan} proposed a DQA for Simon's problem. Xiao et al. \cite{DQA_Xiao} designed a DQA to split Shor's algorithm into two computational nodes. More recently, Li et al. proposed two more distributed versions of Grover's algorithm \cite{DQA_Li} and a Distributed Generalized Simon's problem \cite{DQA_Li_2}. In particular, Qiu et al. \cite{DQA_ec} proposed a universal error-correction method for distributed quantum computing.

In our work, inspired by DQA, we propose a distributed algorithm for DLPs. To this end, we develop an approach that tests whether the solution lies in a given set. Moreover, we obtain the solution to the DLP by intersecting sets that contain the solution. In our algorithm, the success probability is lower bounded by $\mathcal{P}'= \Omega(e^{-\frac{p}{2^p}})$, where $p$ is the number of iterations. Let $m$ denote the size of the first register and let $n<m-1$ denote the size of the second register. The space complexity is $O(2m+n)$, which is smaller than $O(3m)$, the space complexity of Shor's discrete logarithm algorithm. Moreover, the time complexity to find the final solution with $K$ QPUs is $O(m^5n\sqrt{2^n}/K)$. Finally, our algorithm requires only classical communication, with classical communication complexity $O(\log_2(r\log_2 r))$, where $r$ is the (multiplicative) order of the integer $a$ modulo $N$, i.e., the smallest positive integer such that $a^r = 1 \pmod{N}$. The major contributions of this paper are summarized as follows:
\begin{itemize}
    \item We design a distributed quantum discrete logarithm algorithm by testing whether the solution is contained in a given set.
    \item We develop a dichotomy-like strategy that determines the final solution by iteratively shrinking the candidate set.
    \item We analyze the properties of our algorithm and show that, compared with the original algorithm, our distributed algorithm can reduce the qubit count and increase the success probability without any quantum communication.
\end{itemize}

The remainder of this paper is organized as follows: The discrete logarithm problem and its corresponding algorithm are introduced in Section 2. In Section 3, we introduce our proposed algorithm and a strategy to solve DLPs with our algorithm. After this, the correctness and success probability of our algorithm are shown in Section 4. Finally, the time, space, and communication complexities of our algorithm are analyzed in Section 5. Then, the numerical results are summarized in Section 6, and a short conclusion is presented in Section 7.

\section{Preliminary}
In this section, we introduce the Discrete Logarithm Problem (DLP) and the original algorithm to solve it. The discrete logarithm problem can be described as follows:
\begin{definition}
    Given $a,b,r,N \in \mathbb{Z}^+$ with $a^r=1 \pmod N$ and there exists $t \in \mathbb{Z}^+$ such that $b=a^t \pmod N$, then the problem of finding integer $t$ is called \textbf{Discrete Logarithm Problem}.
    \label{Former_Total_def}
\end{definition}
To solve the quantum discrete logarithm algorithm, the property of the function $f(x_1, x_2) = a^{x_1}b^{x_2} \pmod N$ follows: 
\begin{fact}
    Given $a,b,r,N \in \mathbb{Z}^+$ with $a^r=1 \pmod N$, $\exists t \in \mathbb{Z}^+$ such that $b=a^t \pmod N$, and $f(x_1, x_2) = a^{x_1}b^{x_2} \pmod N$, then $t$ satisfies that $\forall l \in \mathbb{Z}^+, f(x_1-tl, x_2+l) = f(x_1, x_2) \pmod N$.
    \label{function_origin}
\end{fact}
According to this fact, the quantum circuit of Shor's quantum algorithm \cite{DLP_first} for solving the discrete logarithm problem is shown in Figure \ref{DicreteLog_original}. This algorithm contains 3 registers, in which the upper two registers are applied to preserve obtained quantum information, while the third contains the ancilla qubits. For $m$-bit integer $0\leq z < N$, gates
\begin{equation}
\begin{split}
    M_a|z\rangle=|za \pmod{N}\rangle,\\
    M_b|z\rangle=|zb \pmod{N}\rangle,
\end{split}
\end{equation}
which are controlled by the first and second registers are applied to the third register and obtain
\begin{align}
    \Lambda(M_{a})&|x\rangle|z\rangle=|x\rangle M^x_{a}|z\rangle = |x\rangle|za^x \pmod N\rangle\\
    \Lambda(M_{b})&|y\rangle|z\rangle= |y\rangle M^y_{b}|z\rangle = |y\rangle|zb^y \pmod N\rangle,
\end{align}
where $x$ and $y$ are integers. Notice that the gates $M_{a}$ and $M_{b}$ share the same eigenvectors \cite{Kitaev95}
\begin{equation}
    |\psi_l\rangle = \frac{1}{\sqrt{r}} \sum^{r-1}_{k=0} \omega_r^{-lk}|a^k \pmod N\rangle,
\end{equation}
where $\omega^l_r = e^{2 \pi i l/ r}$, $i = \sqrt{-1}$, and $|a^k \pmod N\rangle$ is a $m$-bit decimal number. To simplify, we denote $\alpha^l = \omega^l_r$ and $\beta^l = \omega^{tl}_r$ as the eigenvalues which corresponding to $|\psi_l\rangle$ and omit $|a^k \pmod N\rangle = |a^k\rangle$ in the following. Also, the eigenvectors follow
\begin{equation}
    \frac{1}{\sqrt{r}}\sum_{l=0}^{r-1} |\psi_l\rangle = |0^{m-1}1\rangle,
\end{equation}
where $0^{m-1}1=00...01$ is decimal 1 which is represented with $m$-bit binary string.

After applying the controlled gates, with a large enough $m$ (e.g. $m=\lceil \log_2 r \rceil+\log_2 \frac{1}{\epsilon}$ suggested in \cite{QCQI}), the quantum states in the upper two quantum states become
\begin{equation}
    \begin{split}
    \Lambda^{2:3}(M_{b}) \Lambda^{1:3}(M_{a}) \frac{1}{2^m}  \sum_{x=0}^{2^m-1}| x\rangle \sum_{y=0}^{2^m-1}|y\rangle |0^{m-1}1\rangle
    =&\frac{1}{2^m \sqrt{r}} \sum_{x=0}^{2^m-1}| x\rangle \sum_{y=0}^{2^m-1}|y\rangle \sum_{l=0}^{r-1} \alpha^{lx} \beta^{ly} |\psi_l\rangle \\
    =&\frac{1}{2^m \sqrt{r}} \sum_{l=0}^{r-1} \sum_{x=0}^{2^m-1} e^{2 \pi i l x/ r}| x\rangle \sum_{y=0}^{2^m-1} e^{2 \pi i l t y/ r} |y\rangle  |\psi_l\rangle,
    \end{split}
\end{equation}
where $\Lambda^{i:j}$ represents that the control register is $i$ and gate is applied on register $j$ while acts trivially on others. After applying the inverse Quantum Fourier Transform $QFT^\dagger$, the quantum state becomes
\begin{equation}
\begin{split}
    &(QFT_{2^m}^\dagger \otimes QFT_{2^m}^\dagger \otimes I^{\otimes m}) \frac{1}{2^m \sqrt{r}} \sum_{l=0}^{r-1} \sum_{x=0}^{2^m-1} e^{2 \pi i l x/ r}| x\rangle \sum_{y=0}^{2^m-1} e^{2 \pi i l t y/ r} |y\rangle |\psi_l\rangle\\
    = &\frac{1}{\sqrt{r}}\sum_{l=0}^{r-1} |\widetilde{l/r}\rangle |\widetilde{tl/r}\rangle |\psi_l\rangle
\end{split}
\end{equation}
where $|\widetilde{l/r}\rangle|\widetilde{tl/r}\rangle$ is the approximation of $|l/r\rangle|tl/r\rangle$. Thus, the upper two registers are measured and we obtain classical value $\widetilde{l/r}$ and $\widetilde{tl/r}$. Finally, we obtain $t = \widetilde{tl/r}\left(\widetilde{l/r}\right)^{-1}$.

\begin{figure*}[h!]
    \centering
    \includegraphics[scale=0.8]{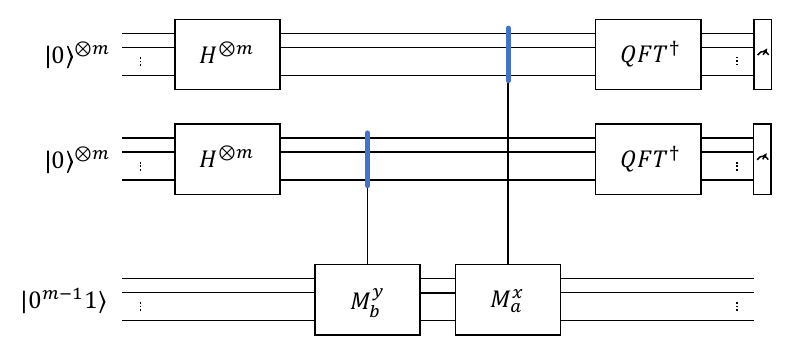}
    \caption{Original quantum discrete logarithm algorithm.}
    \label{DicreteLog_original}
\end{figure*}

\section{Our Algorithm}
The details of our distributed discrete logarithm algorithm are shown in this section. In general, the purpose of our algorithm is to determine whether the solution for the distributed discrete logarithm problem is contained in a specific set. To this approach, we first discuss how the combination of the eigenvectors will change if we remove some eigenvectors and how our algorithm removes specific eigenvectors. Finally, a brief process of our algorithm is organized at the end of the section.

\subsection{Quantum gates}
Before introducing our algorithm, we first introduce the quantum gates applied in our algorithm. The Hadamard gate can be represented as
\begin{equation}
    H = \frac{1}{\sqrt{2}}
    \begin{pmatrix}
    1   &1\\
    1   &-1\\
    \end{pmatrix}.
\end{equation}
And $QFT_K$ is the quantum Fourier transform
\begin{equation}
    QFT_K: |j\rangle \to \frac{1}{\sqrt{K}}\sum^{K-1}_k e^{2\pi i j k/K}|k\rangle \quad (0\leq j < K).
\end{equation}
Furthermore, for $0 \leq z < N$ control gates are also utilized in our algorithm:
\begin{align}
    \Lambda(M_b)|x\rangle|z\rangle &=|x\rangle|z b^x \pmod{N} \rangle, \\
    \Gamma_{\tau}(M_a)|x\rangle|y\rangle|z\rangle
    &=|x\rangle|y\rangle |z a^{-(x+\tau)y} \pmod{N} \rangle,
\end{align}
where $\tau \in \mathbb{Z}$. Actually, our newly defined gate $\Gamma_{\tau}(M_a)$ can be easily applied with existing gates $\Lambda(M_a)$ and $\Lambda(M_{a^\tau})$, 
and $\Xi(U):|s\rangle|x\rangle|z\rangle \to |s\rangle U^s(|x\rangle|z\rangle)$ via
\begin{equation}
    \Gamma_{\tau}(M^\dagger_{a})=(I^{\otimes n} \otimes \Lambda(M^\dagger_{a^\tau}) \Xi(\Lambda(M^\dagger_a)),
\end{equation}
which we will prove in the subsequent section.

For any $s,x \in \mathbb{N}$, with the combination of gates $\Lambda(M_b)$ and $\Gamma_{\tau}(M_a)$, we have
\begin{equation} \label{fucntion_state}
\begin{split}
    \Lambda^{2:3}(M_{b})\Gamma^{1,2:3}_{\tau}(M_a)|s\rangle |x\rangle |0^{m-1}1\rangle
    = &\Lambda^{2:3}(M_{b})|s\rangle |x\rangle |a^{-(s+\tau)x} \pmod N \rangle \\ 
    = &|s\rangle |x\rangle |a^{-(s+\tau)x} b^{x} \pmod N \rangle \\
    = &|s\rangle |x\rangle |\hat{f}_{s+\tau}(x)\rangle,
\end{split}
\end{equation}
where $\hat{f}_s(x) = a^{-sx}b^{x} \pmod N$. Moreover, another controlled gate is involved to realize the $g(x)$, which will be further introduced in the next part:
\begin{equation}
    U_{g}: |x\rangle |0\rangle \mapsto |x\rangle (\sqrt{1-g^2(x)}|0\rangle+g(x)|1\rangle),
\end{equation}
where $g:\{0,1\}^m\to\{0,1\}$, $m$ is the length of $x$, and 
\begin{equation}
    g(x)=\begin{cases}
        1, &x=0 \\
        0, &x=1,2,...2^m-1.
    \end{cases}
\end{equation}
This controlled gate helps us to distinguish the kept eigenvectors from others by setting the corresponding ancilla qubit to $|1\rangle$.

\subsection{The superposition of eigenvectors}
According to the former introduced Def. \ref{Former_Total_def}, a function $\hat{f}_s(x) = a^{-sx}b^{x} \pmod N$ can be applied to solve the DLPs. It can be recognized that
\begin{equation}
\begin{split}
    \Gamma_{\tau}(M_a) |s\rangle|x\rangle |\psi_l\rangle &= \Gamma_{\tau}(M_a) |s\rangle|x\rangle \frac{1}{\sqrt{r}} \sum^{r-1}_{k=0} \omega_r^{-lk}|a^k\rangle \\
    &= |s\rangle|x\rangle \frac{1}{\sqrt{r}} \sum^{r-1}_{k=0} \omega_r^{-lk}|a^{k+(s+\tau)x} \pmod{N}\rangle= \omega^{-(s+\tau)xl}_r|s\rangle|x\rangle |\psi_l\rangle,\\
    \Lambda(M_b) |x\rangle |\psi_l\rangle &= \Lambda(M_b) |x\rangle \frac{1}{\sqrt{r}} \sum^{r-1}_{k=0} \omega_r^{-lk}|a^k\rangle\\
    &= |x\rangle \frac{1}{\sqrt{r}} \sum^{r-1}_{k=0} \omega_r^{-lk}|a^k b^x \pmod{N}\rangle = \omega_r^{xl} |x\rangle |\psi_l\rangle.
\end{split}
\end{equation}
The eigenvectors of $\Gamma_{\tau}(M_a)$ and $\Lambda(M_b)$ are
\begin{equation} \label{eigenvectors}
\begin{split}
    |s\rangle|x\rangle |\psi_l\rangle &= |s\rangle|x\rangle\frac{1}{\sqrt{r}} \sum^{r-1}_{k=0} \omega_r^{-lk}|a^k\rangle,\\
    |x\rangle |\psi_l\rangle &= |x\rangle \frac{1}{\sqrt{r}} \sum^{r-1}_{k=0} \omega_r^{-lk}|a^k\rangle,
\end{split}
\end{equation}
where $\omega^l_r = e^{2 \pi i l/ r}$, $k\in\mathbb{N}$ and $k<r$. 

Because the quantum state $|\hat{f}_s(x)\rangle$ can be described as the superposition state of the eigenvectors. Also, knowing the superposition state of the eigenvectors is $|0^{m-1}1\rangle$. According to Eq. \ref{fucntion_state}, for a known $s$, we write the quantum state
\begin{equation} \label{function_final}
\begin{split}
    \frac{1}{\sqrt{2^m}}\sum^{2^m-1}_{x=0} |s\rangle|x\rangle|\hat{f}_{s+\tau}(x)\rangle =& \frac{1}{\sqrt{2^m}} \sum^{2^m-1}_{x=0} \Lambda(M_b)\Gamma_{\tau}(M_a)|s\rangle|x\rangle|0^{m-1}1\rangle \\
    =& \frac{1}{\sqrt{2^m}} \sum^{2^m-1}_{x=0} \Lambda^{2:3}(M_b)\Gamma^{1,2:3}_{\tau}(M_a)|s\rangle|x\rangle \frac{1}{\sqrt{r}}\sum_{l=0}^{r-1} |\psi_l\rangle \\
    =& \frac{1}{\sqrt{2^mr}}\sum^{r-1}_{l=0} \sum^{2^m-1}_{x=0} e^{2 \pi i (t-\tau-s) l x/ r} |s\rangle |x\rangle |\psi_l\rangle.
\end{split}
\end{equation}

As mentioned above, the purpose of our algorithm is to distinguish the set which contains the solution. Thus, we aim to make the quantum states corresponding to the solution different from others. If $\hat{f}_s(x)$ is a constant function then the corresponding $s$ is the required answer because $\hat{f}_t(x)=1 \pmod{N}$. And the Fourier Transform of a constant function is $0$. Thus, we apply the function $g(x)$ to distinguish whether the result is $0$.

Then, for a constant integer $\tau$ which represents a start point and $2^n<r$ as the size of set, we denote 
\begin{equation}
    S_{n,\tau} = \{\tau+s \pmod{r}| s= 0,1,2,...,2^n-1 \}
\end{equation}
and expression 
\begin{equation} \label{double_combine}
    \sum_{s \in S_{n,\tau}} \sum_{l=0}^{r-1} (-1)^s g((t-s)l) |\psi_l\rangle
\end{equation}
represents a superposition state with eigenvectors of $\hat{f}_s(x)$ with different elements in set $S_{n,\tau}$. For further analysis, we conduct the following lemma to describe the property of the upper expression. 
\begin{lemma} \label{MAIN}
    Given set $S_{n,\tau} = \{\tau+s \pmod{r}| s= 0,1,2,...,2^n-1 \}$ where $2^n<r$ and $\tau$ is integer. Then, if $t \in S_{n,\tau}$, 
    \begin{equation}
        \sum_{s \in S_{n,\tau}} \sum_{l=0}^{r-1} (-1)^s g((t-s)l) |\psi_l\rangle = \frac{(-1)^{t+1}}{\sqrt{r}} \sum_{k=1}^{r-1} |a^k\rangle + \left(\frac{(-1)^{t+1}}{\sqrt{r}}+(-1)^t \sqrt{r}\right) |0^{m-1}1\rangle.
    \end{equation}
    Otherwise, if $t \notin S_{n,\tau}$, $\sum_{s \in S_{n,\tau}} \sum_{l=0}^{r-1} (-1)^s g((t-s)l) |\psi_l\rangle =0$.
\end{lemma}
\begin{proof}
    For expression 
    \begin{equation}
        \sum_{s \in S_{n,\tau}} \sum_{l=0}^{r-1} (-1)^s g((t-s)l) |\psi_l\rangle,
    \end{equation}
    if there is $s \in S_{n,\tau}$ such that $s=t \pmod{r}$, the expression can be written as 
    \begin{equation}
    \begin{split}
        &\sum_{s \in S_{n,\tau}} \sum_{l=0}^{r-1} (-1)^s g((t-s)l) |\psi_l\rangle\\
        =& \sum_{s \in S_{n,\tau}, s \neq t} \sum_{l=0}^{r-1} (-1)^s g((t-s)l) |\psi_l\rangle + \sum_{l=0}^{r-1} (-1)^t |\psi_l\rangle\\
        =& \sum_{s \in S_{n,\tau}, s \neq t} (-1)^s |\psi_0\rangle + \sum_{l=0}^{r-1} (-1)^t |\psi_l\rangle\\
        =& \frac{(-1)^{t+1}}{\sqrt{r}} \sum_{k=1}^{r-1} |a^k\rangle + \left(\frac{(-1)^{t+1}}{\sqrt{r}}+(-1)^t \sqrt{r}\right) |0^{m-1}1\rangle.
    \end{split}
    \end{equation}
    And if $t \notin S_{n,\tau}$, we have 
    \begin{equation}
    \begin{split}
        &\sum_{s \in S_{n,\tau}} \sum_{l=0}^{r-1} (-1)^s g((t-s)l) |\psi_l\rangle\\
        =& \sum_{s \in S_{n,\tau}} \sum_{l=1}^{r-1} (-1)^s g((t-s)l) |\psi_l\rangle + \sum_{s \in S_{n,\tau}} (-1)^s g((t-s)l) |\psi_l\rangle\\
        =& \sum_{s \in S_{n,\tau}} (-1)^s |\psi_l\rangle = 0.
    \end{split}
    \end{equation}
\end{proof}
It can be recognized that if $t \in S_{n,\tau}$, the state $|0^{m-1}1\rangle$ has a significant proportion. Therefore, with the Lemma \ref{MAIN}, we can determine whether $t \in S_{n,\tau}$ according to whether we can observe $1$ from the Superposition state eigenvectors with high probability. 

\subsection{Details of our algorithm}
We have introduced how to distinguish whether $t \in S_{n,\tau}$ according to the value of the superposition state eigenvectors. In this part, we exhibit the details of our algorithm. 

Firstly, we initialize the quantum state as
\begin{equation}
    |\phi_0\rangle = |0\rangle^{\otimes n}|0\rangle^{\otimes m}|0^{m-1}1\rangle|0\rangle,
\end{equation}
where the first and second registers contain $n$ and $m$ qubits separately. The third register also includes $m$ qubits, the initial state $|0^{m-1}1\rangle$ represents the $m$-bit decimal $1$ whose every qubit is initialized as $|0\rangle$ except the lowest qubit as $|1\rangle$. The last register contains $1$ ancilla qubit, which is applied to mark whether the eigenvectors are required. 

Then, we apply Hadamard gates to each of the qubits in the first and second registers:
\begin{equation}
    |\phi_1\rangle= \frac{1}{\sqrt{2^m 2^n}}\sum_{s=0}^{2^{n}-1} |s\rangle \sum_{x=0}^{2^m-1} |x\rangle |0^{m-1}1\rangle |0\rangle.
\end{equation}
Then, after applying $\Gamma_{\tau}(M_{a})$ and $\Lambda(M_{b})$ according to the first and second registers, then the final state becomes
\begin{equation}
\begin{split}
    |\phi_3\rangle=& \Lambda^{2:3}(M_{b})|\phi_2\rangle= \Lambda^{2:3}(M_{b})\Gamma^{1,2:3}_{\tau}(M_{a})|\phi_1\rangle\\
    =&\Lambda^{2:3}(M_{b})\Gamma^{1,2:3}_{\tau}(M_{a})\frac{1}{\sqrt{2^m 2^n}}\sum_{s=0}^{2^{n}-1} |s\rangle \sum_{x=0}^{2^m-1} |x\rangle |0^{m-1}1\rangle |0\rangle\\
    =&\frac{1}{\sqrt{ 2^{n+m}r}}\sum_{s=0}^{2^{n}-1} |s\rangle \sum_{l=0}^{r-1} \sum_{x=0}^{2^m-1} e^{2 \pi i (t-\tau-s) l x/ r} |x\rangle |\psi_l\rangle |0\rangle.
\end{split}
\end{equation}
where $\tau \in \{0,1,...,r-1\}$ is a manually chosen integer which is determine the set $S_{n,\tau}$ and $\Gamma^{1,2:3}_{\tau}(M_{a})$ represents that the gate is controlled with the register 1 and 2 and applied on the register 3. In the next step, before we analyze the error in the next section, we assume that we can exactly obtain $(t-\tau-s)l \pmod{r}$ via the $QFT_{2^m}^\dagger$. Thus, the $QFT_{2^m}^\dagger$ leads to 
\begin{equation}
    |\phi_4\rangle= (I^{\otimes n}\otimes QFT_{2^m}^\dagger \otimes I^{\otimes m} \otimes I)|\phi_3\rangle=
    \frac{1}{\sqrt{2^{n}r}}\sum_{s=0}^{2^{n}-1} |s\rangle \sum_{l=0}^{r-1} |(t-\tau-s)l \pmod{r} \rangle |\psi_l\rangle |0\rangle.
\end{equation}
Then, $U_{g}$ acts on the second and fourth registers to determine whether the eigenvector should be kept:
\begin{equation}
    |\phi_5\rangle= \frac{1}{\sqrt{2^{n}r}}\sum_{s=0}^{2^{n}-1} |s\rangle \sum_{l=0}^{r-1} |(t-\tau-s)l\rangle |\psi_l\rangle (\sqrt{1-g^2((t-\tau-s)l)}|0\rangle+g((t-\tau-s)l)|1\rangle),
\end{equation}
where we omitted $\bmod$ for short. In the following steps, we aim to release the dependency relationship between the third and the first two registers by uncomputing the above steps. Therefore, we apply the $QFT_{2^m}$ to the second register and obtain
\begin{equation}
    |\phi_6\rangle=\frac{1}{\sqrt{2^{n+m}r}}\sum^{r-1}_{l=0} \sum_{s=0}^{2^{n}-1} \sum^{2^m-1}_{x=0} |s\rangle |x\rangle e^{2 \pi i (t-\tau-s) lx/r} |\psi_l\rangle (\sqrt{1-g^2((t-\tau-s)l)}|0\rangle+g((t-\tau-s)l)|1\rangle).
\end{equation}
Similarly, we also apply the inverse of $\Gamma_{\tau}(M_{a})$ and $\Lambda(M_{b})$:
\begin{equation}
    |\phi_8\rangle=\frac{1}{\sqrt{2^{n+m}r}}\sum^{r-1}_{l=0} \sum_{s=0}^{2^{n}-1} \sum^{2^m-1}_{x=0} |s\rangle |x\rangle |\psi_l\rangle(\sqrt{1-g^2((t-\tau-s)l)}|0\rangle+g((t-\tau-s)l)|1\rangle).
\end{equation}
To cancel the dependency of the first register, Hadamard gates are applied on each qubit of the first and second registers, and the quantum state becomes 
\begin{equation} \label{quantum_state_9}
\begin{split}
    |\phi_9\rangle =&\frac{1}{\sqrt{2^{2n}r}}\sum^{r-1}_{l=0} \sum_{s=0}^{2^{n}-1} \sum_{s'=0}^{2^{n}-1} (-1)^{s\cdot s'} |s'\rangle |0\rangle |\psi_l\rangle(\sqrt{1-g^2((t-\tau-s)l)}|0\rangle+g((t-\tau-s)l)|1\rangle)\\
    =& \frac{1}{\sqrt{2^{2n}r}} \sum_{s'=0}^{2^{n}-1} |s'\rangle |0\rangle \sum_{s=0}^{2^{n}-1} \sum^{r-1}_{l=0} (-1)^{s\cdot s'}  |\psi_l\rangle(\sqrt{1-g^2((t-\tau-s)l)}|0\rangle+g((t-\tau-s)l)|1\rangle),
\end{split}
\end{equation}
where $\cdot$ represents the inner product of two binary vectors. Knowing that the value of $g((t-\tau-s)l)$ depends on the value of $(t-\tau-s)l \pmod{r}$, if and only if $t=\tau+s \pmod{r}$ or $l=0 \pmod r$, $g((t-\tau-s)l)=0$. And the total process is shown in Fig. \ref{Total_process}. Thus, we discuss both cases of whether there is an $s$ such that $t=\tau+s \pmod{r}$.

\begin{figure}
    \centering
    \includegraphics[scale=0.7]{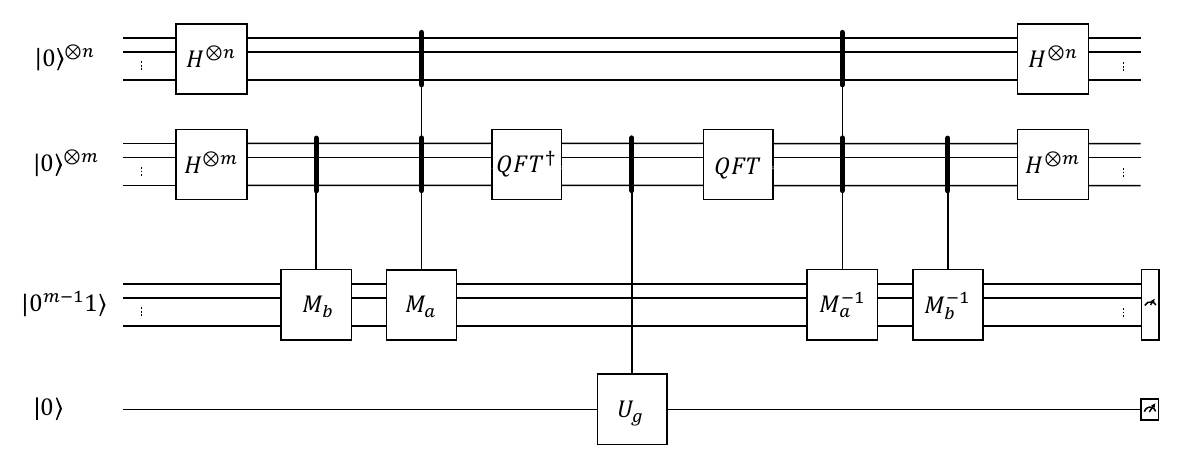}
    \caption{The total process of our Algorithm \ref{DQDLA}.}
    \label{Total_process}
\end{figure}

\subsection{Probability of observing $0^{m-1}1$ from the third register}
Considering that $|\phi_9\rangle$ of the last register is $|1\rangle$, the unnormalized quantum state is
\begin{equation} \label{before_s}
\begin{split}
    \frac{1}{\sqrt{2^{2n}r}} \sum_{s'=0}^{2^{n}-1} |s'\rangle |0\rangle \left(\sum_{s=0}^{2^{n}-1} \sum^{r-1}_{l=0} (-1)^{s\cdot s'} g((t-\tau-s)l) |\psi_l\rangle\right)|1\rangle.
\end{split}
\end{equation}

If there is $s$ such that $t=\tau+s \pmod{r}$ and $s'=0$, then following Lemma \ref{MAIN}, we have
\begin{equation} \label{state_in_1_0_4_1}
\begin{split}
    \frac{1}{2^{n}r} |0\rangle |0\rangle &\left ( \sum^{r-1}_{k=1} (2^n-1) |a^k\rangle 
    + \left( r + (2^n-1) \right) |0^{m-1}1\rangle \right ) |1\rangle .
\end{split}
\end{equation}

Otherwise, if $s' \neq 0$, also following Lemma \ref{MAIN}, the state is
\begin{equation} \label{state_in_1_1_4_1}
\begin{split}
    \frac{1}{2^{n}r}  \sum_{s'=1}^{2^{n}-1} |s'\rangle |0\rangle &\left ( \sum^{r-1}_{k=1} (-1)^{(t-\tau)\cdot s'+1} |a^k\rangle + (-1)^{(t-\tau)\cdot s'}(r-1) |0^{m-1}1\rangle \right ) |1\rangle .
\end{split}
\end{equation}

According to this quantum state, we have:
\begin{lemma}
    For quantum state $|\phi_9\rangle$, if $t\in S_{n,\tau} = \{\tau+s \pmod{r}| s= 0,1,2,...,2^n-1 \}$, the probability of the fourth register is 1 is
    \begin{equation}
        \frac{(2^n-1)^2(r-1)+(r+2^n-1)^2+(2^n-1)(r-1)+(2^n-1)(r-1)^2}{2^{2n}r^2}.
    \end{equation}
\end{lemma}
\begin{proof}
    The probability of the first register is 0 and the fourth register is 1, which can be represented as
    \begin{equation}
    \begin{split}
    &\left| \frac{1}{2^{n}r} \langle 0| \langle 0| \left ( \sum^{r-1}_{k=1} (2^n-1) \langle a^k| \langle 1| + \left( r + (2^n-1) \right) \langle 1| \langle 1| \right ) \right.\\
    &\left. \frac{1}{2^{n}r} |0\rangle |0\rangle \left ( \sum^{r-1}_{k=1} (2^n-1) |a^k\rangle |1\rangle + \left( r + (2^n-1) \right) |0^{m-1}1\rangle |1\rangle \right ) \right |\\
    =& \frac{1}{2^{2n}r^2} \left (\sum^{r-1}_{k=1} (2^n-1)^2 + (r-(2^n-1))^2 \right ) 
    = \frac{(r-1)(2^n-1)^2 + (r-2^n-1)^2}{2^{2n}r^2}.
    \end{split}
    \end{equation}
    Then the probability of the first register is NOT 0 and the fourth register is 1
    \begin{equation}
    \begin{split}
        &\left | \frac{1}{2^{n}r}  \sum_{s'=1}^{2^{n}-1} \langle s'| \langle 0| \left ( \sum^{r-1}_{k=1} (-1)^{(t-\tau)\cdot s'+1} \langle a^k| \langle 1| + (-1)^{(t-\tau)\cdot s'}(r-1) \langle 1| \langle 1| \right ) \right .
        \\ & \left. \frac{1}{2^{n}r}  \sum_{s'=1}^{2^{n}-1} |s'\rangle |0\rangle \left ( \sum^{r-1}_{k=1} (-1)^{(t-\tau)\cdot s'+1} |a^k\rangle |1\rangle + (-1)^{(t-\tau)\cdot s'}(r-1) |0^{m-1}1\rangle |1\rangle \right ) \right | \\
        = &\frac{1}{2^{2n}r^2} \sum_{s'=1}^{2^{n}-1} \left (\sum^{r-1}_{k=1} 1 + (r-1)^2 \right )
        = \frac{(2^n-1)(r-1)+(2^n-1)(r-1)^2}{2^{2n} r^2}.
    \end{split}
    \end{equation}
    Therefore, the probability of the fourth register is 1 equals 
    \begin{equation}
    \begin{split}
        &\frac{(r-1)(2^n-1)^2 + (r-2^n-1)^2}{2^{2n}r^2} + \frac{(2^n-1)(r-1)+(2^n-1)(r-1)^2}{2^{2n} r^2}\\
        =&\frac{(2^n-1)^2(r-1)+(r+2^n-1)^2+(2^n-1)(r-1)+(2^n-1)(r-1)^2}{2^{2n}r^2}
    \end{split}
    \end{equation}
\end{proof}

And, the probability of both the fourth and third registers are 1 follows
\begin{lemma}
    For quantum state $|\phi_9\rangle$, if $t\in S_{n,\tau} = \{\tau+s \pmod{r}| s= 0,1,2,...,2^n-1 \}$, the probability of the fourth register is $1$ and the third register are $0^{m-1}1$ is
    \begin{equation}
        \frac{(r+2^n-1)^2+(2^n-1)(r-1)^2}{2^{2n}r^2}.
    \end{equation}
\end{lemma}
\begin{proof}
    The probability of all the fourth, third, and the registers are 1 equals
    \begin{equation}
    \begin{split}
        &\left| \left(\frac{1}{2^{n}r} \langle 0| \langle 0|  \left( r + (2^n-1) \right) \langle 0^{m-1}1| \langle 1| \right) \left(\frac{1}{2^{n}r} |0\rangle |0\rangle \left( r + (2^n-1) \right) |0^{m-1}1\rangle |1\rangle \right)\right |\\
        =&\frac{(r+2^n-1)^2}{2^{2n}r^2},
    \end{split}
    \end{equation}
    and probability of (1) the fourth and third registers are 1, and (2) the first register is NOT 1 equals
    \begin{equation}
    \begin{split}
        &\left | \left(\frac{1}{2^{n}r}  \sum_{s'=1}^{2^{n}-1} \langle s'| \langle 0| (r-1) \langle 0^{m-1}1| \langle 1|\right)\left( \frac{1}{2^{n}r}  \sum_{s'=1}^{2^{n}-1} |s'\rangle |0\rangle (r-1) |0^{m-1}1\rangle |1\rangle \right) \right | \\
        =& \frac{(2^n-1)(r-1)^2}{2^{2n}r^2}.
    \end{split}
    \end{equation}
    Then the probability can be obtained by adding both above-mentioned results.
\end{proof}

Combining the upper results, we utilize a proposition to describe the probability of observing $1$ from the third register 
\begin{proposition} \label{t_in}
    For quantum state $|\phi_9\rangle$, if $t\in S_{n,\tau} = \{\tau+s \pmod{r}| s= 0,1,2,...,2^n-1 \}$ for integer $\tau$, the probability of the fourth register is $1$ is no less than $\frac{1}{2^n}$. If the fourth register is $1$, then the probability of observing $0^{m-1}1$ from the third register is
    \begin{equation}
        \frac{r^2+(2^n-1)^2+(2^n)(r^2+1)}{(2^n-1)^2r+r^2+(2^n-1)r(r+1)}>\frac{1}{2+\frac{2}{r}}.
    \end{equation}
\end{proposition}
\begin{proof}
According to the above-mentioned lemmas, it is clear that the probability equals the probability of the third register is $0^{m-1}1$ and the fourth register is $1$ divided by the probability of the fourth register is $1$, which is
    \begin{equation}
    \begin{split}
        &\frac{(r+2^n-1)^2+(2^n-1)(r-1)^2}{(2^n-1)^2(r-1)+(r+2^n-1)+(2^n-1)(r-1)+(2^n-1)(r-1)^2}\\
        =& \frac{r^2+(2^n-1)^2+(2^n-1)(r^2+1)}{(2^n-1)^2r+r^2+(2^n-1)r(r+1)}
        >\frac{r^2+(2^n-1)^2+(2^n-1)(r^2+1)}{(2^n-1)^2r+2^nr(r+1)}.
    \end{split}
    \end{equation}
It's proper to set $2^n-1<r-1$ because the element number of the set would not be larger than $r$. Therefore, the above expression is greater than
    \begin{equation}
    \begin{split}
        &\frac{r^2+(2^n-1)^2+(2^n-1)(r^2+1)}{(2^n-1)r(r+1)+2^nr(r+1)}\\
        >&\frac{2^n r^2}{(2^{n+1}-1)r^2+(2^{n+1}-1)r}
        >\frac{1}{2+\frac{2}{r}}.
    \end{split}
    \end{equation}
\end{proof}

Then, considering the case of $t \notin S_{n,\tau}$, the quantum state $|\phi_9\rangle$ with $1$ as the last quantum register follows
\begin{proposition} \label{t_notin}
    For quantum state $|\phi_9\rangle$, if $t\notin S_{n,\tau} = \{\tau+s \pmod{r}| s= 0,1,2,...,2^n-1 \}$ for integer $\tau$, the probability of the fourth register is $1$ is no less than $1/r$. And if the fourth register is $1$, then the probability of observing $a^k \pmod{N}$ from the third register is $\frac{1}{r}$. Especially, the probability of observing $0^{m-1}1$ is $\frac{1}{r}$. 
\end{proposition}
\begin{proof}
    According to the quantum state \ref{quantum_state_9}, if there is NOT $s$ such that $t=\tau+s \pmod{r}$, we have
    \begin{equation} 
    \begin{split}
        |\phi_9\rangle = \frac{1}{\sqrt{2^{2n}r}} \sum_{s'=0}^{2^{n}-1} |s'\rangle |0\rangle \sum_{s=0}^{2^{n}-1} (-1)^{s\cdot s'} \left (\sum^{r-1}_{l=1} |\psi_l\rangle|0\rangle+|\psi_0\rangle|1\rangle \right ),
    \end{split}
    \end{equation}
    Then we split the quantum state whose third register is $|1\rangle$, then 
    \begin{equation}
    \begin{split}
        |\phi_9\rangle = \frac{1}{r} |0\rangle |0\rangle \left (\sum^{r-1}_{k=1} (-1)|a^k\rangle |0\rangle +(r-1)|0^{m-1}1\rangle |0\rangle + \sum^{r-1}_{k=0} |a^k\rangle |1\rangle \right ),
    \end{split}
    \end{equation}
    where the quantum state with the fourth quantum register is $1$ follows:
    \begin{equation} \label{state_notin_1_0_4_1}
    \begin{split}
        \frac{1}{r} |0\rangle |0\rangle\sum^{r-1}_{k=0} |a^k\rangle |1\rangle.
    \end{split}
    \end{equation}
    Therefore, only considering the case of the fourth register is $1$, it can be recognized that the probability of observing $a^k \mod{r}$ is equal for each $k$. Thus the probability of the fourth register is $1$ equals
    \begin{equation}
    \begin{split}
        &\left |\left(\frac{1}{r} \langle 0| \langle 0|\sum^{r-1}_{k=0} \langle a^k| \langle 1|\right) \left( \frac{1}{r} |0\rangle |0\rangle\sum^{r-1}_{k=0} |a^k\rangle |1\rangle \right) \right |
        = \frac{1}{r^2} \sum^{r-1}_{k=0} 1 = \frac{1}{r}.
    \end{split}
    \end{equation}
    For each $a^k$, the probability equals $1/r^2$. Thus, if we know the fourth register is $1$, the probability of each $a^k$ is $1/r$.
\end{proof}

Thus, as shown in Algorithm \ref{DQDLA}, we conduct an algorithm that can return whether a $t$ is included in some specific set. Then in the next section, we introduce how to utilize the Algorithm \ref{DQDLA} to find $t$.

\begin{algorithm}
\caption{The Distributed Quantum Discrete Logarithm Algorithm.}\label{DQDLA}
\SetAlgoLined

\KwIn{Initialize Hadamard gate $H$, Controlled gates $\Gamma_{\tau}(M_a)$ and $\Lambda(M_b)$, $QFT_{2^m}$, and $U_{g}$.}
\KwIn{A set $S_{n,\tau}= \{\tau+s|s=0,1,2,...,2^n-1\}$ with start position $\tau$, size of set $n$.}
\KwOut{Whether $t$ in $S_{n,\tau} $.}

Initialize state $|\phi_0 \rangle= |0\rangle^{\otimes n}|0\rangle^{\otimes m}|1\rangle|0\rangle$\label{stp:init}\;
$|\phi_1\rangle = (H^{\otimes n} \otimes H^{\otimes m} \otimes I^{\otimes m+1})|\phi_0\rangle$\;
Apply $\Lambda(M_b)$ to the second and third registers to obtain $|\phi_2\rangle$\;
Apply $\Gamma_{\tau}(M_a)$ to the first, second and third registers to obtain $|\phi_3\rangle$\;
Apply $QFT_{2^m}^\dagger$ to the second register to obtain $|\phi_4\rangle$\;
Apply $U_{g}$ to the second and fourth registers to obtain $|\phi_5\rangle$\;
Apply $QFT_{2^m}$ to the second register to obtain $|\phi_6\rangle$\;
Apply $\Gamma_{\tau}(M^\dagger_a)$ to the second and third register to obtain $|\phi_7\rangle$\;
Apply $\Lambda(M^\dagger_b)$ to the first, second and third registers to obtain $|\phi_8\rangle$\; 
Apply $(H^{\otimes n} \otimes H^{\otimes m} \otimes I^{\otimes m+1})$ to $|\phi_8\rangle$ to obtain $|\phi_9\rangle$\;
If measuring the fourth register and observe $1$, then if measure the third register and observe $0^{m-1}1$ then return True, else return False.
\end{algorithm}

\subsection{Strategy for searching solution} \label{Strategy}
According to the Proposition \ref{t_in} and \ref{t_notin}, the probabilities of observing $1$ from the third register are different depending on whether $t$ is contained in the set $S_{n,\tau} = \{\tau+s \pmod{r}| s= 0,1,2,...,2^n-1 \}$. Therefore, we can determine whether the $t$ is contained in some sets by whether we can observe $1$ from the third register with high probability:

\begin{lemma} \label{Prob_iteration}
    For quantum state $|\phi_9\rangle$, if the fourth register is 1 and we measure the third register for $p$ times, the probability of observing $0^{m-1}1$ for at least $1$ time if $t \in \{\tau+s \pmod{r}| s= 0,1,2,...,2^n-1 \}$ follows
    \begin{equation}
        1-(1-\frac{1}{2+\frac{2}{r}})^p > 1 - \frac{1}{2^p} (\frac{r+2}{r+1})^p.
    \end{equation}
    Else if $t \notin \{\tau+s \pmod{r}| s= 0,1,2,...,2^n-1 \}$, the probability of observing no $1$ follows
    \begin{equation}
        (1-\frac{1}{r})^p.
    \end{equation}
\end{lemma}
\begin{proof}
    With Proposition \ref{t_in} and \ref{t_notin}, this lemma can be proven with the probability of each state.
\end{proof}

\begin{algorithm}
\caption{The Strategy for searching $t$.}\label{algorithm_strategy}
\SetAlgoLined

\KwIn{$a,b,r\in \mathbb{Z}^+$ such that $a^r=1 \pmod{N}$, and $m=\lceil \log_2 r \rceil+\log_2 \frac{1}{\epsilon}$.}
\KwOut{$t$ such that $a^t=b \pmod{N}$.}
Initial $n<m-1$, $\tau=0$\;
Find if $t\in S_{n,\tau} = \{\tau+s|s=0,1,2,...,2^n-1\}$ with Algorithm \ref{DQDLA} and obtain $p$ results \label{stp:eval}\;
If Algorithm \ref{DQDLA} return True for at least $1$ time, then set $n\gets n-1$\, else set $\tau \gets \tau + 2^n$\;
if $n \neq 0$, go back to step \ref{stp:eval} \;
\Return $t=\tau$
\end{algorithm}

If we find a list of sets containing $t$, the solution of $t$ can be obtained with the intersection of sets. To determine which sets contain the $t$, we need to measure $0^{m-1}1$ from the third register at least 1 time. If we observe at least 1 time $1$ from a set, the solution may be contained in this set. Otherwise, if we do not observe any $0^{m-1}1$, we can infer that the solution is not contained in the set and we need to move and check the next possible set (See Algorithm \ref{algorithm_strategy}).

To this approach, the number of measurements $p$ should be determined to promise the probability of observing $1$ for at least $0^{m-1}1$ times from the third register can be sufficient discrimination. More specifically, if we conduct a list of sets to find a $t$ (we need $O(\lceil \log_2 r \rceil)$ sets), the total success probability follows 
\begin{theorem} \label{success_prob}
    For Algorithm \ref{algorithm_strategy}, if the the number of iterations is $p$, then the success probability $\mathcal{P}$ follows
    \begin{equation}
        \mathcal{P}> (1 - \frac{1}{2^p} (\frac{r+2}{r+1})^p)^{\lceil \log_2 r \rceil}>\frac{1}{e^{\frac{2p}{d^p-1}}},
    \end{equation}
    where $p+\log_2 p\leq\log_2 r$ and $d=\frac{2(r+1)}{r+2}$.
\end{theorem}
\begin{proof}
    With Lemma \ref{Prob_iteration}, the probability of observing $0^{m-1}1$ for at least $1$ time increases with $p$ increase, while the probability of observing no $1$ decreases. Thus, we limit
    \begin{equation}
        1 - \frac{1}{2^p} (\frac{r+2}{r+1})^p \leq (1-\frac{1}{r})^p
    \end{equation}
    to avoid decreasing the final probability. With the upper inequality, we have
    \begin{equation}
    \begin{split}
        &1 - \frac{1}{2^p} (\frac{r+2}{r+1})^p \leq (1-\frac{1}{r})^p\\
        \Rightarrow &\log_2 (1 - \frac{1}{2^p} (\frac{r+2}{r+1})^p) \leq \log_2 (1-\frac{1}{r})^p\\
        \Rightarrow &  - \frac{1}{2^p} (\frac{r+2}{r+1})^p \leq -\frac{p}{r}\\
        \Rightarrow & p \log_2(\frac{1}{2} \frac{r+2}{r+1}) \geq \log_2(p) -\log_2(r)\\
        \Rightarrow & p (1-\log_2(\frac{r+2}{r+1})) \leq \log_2(r) -\log_2(p).
    \end{split}
    \end{equation}
    Knowing that $p (1-\log_2(\frac{r+2}{r+1})) < p$, to further limit the range of $p$, we set
    \begin{equation}
        p \leq \log_2(r) -\log_2(p).
    \end{equation}
    As we have $\lceil \log_2 r \rceil$ sets, the total success probability
    \begin{equation}
        \mathcal{P}> (1 - \frac{1}{2^p} (\frac{r+2}{r+1})^p)^{\lceil \log_2 r \rceil}.
    \end{equation}
    Then, we have
    \begin{equation}
    \begin{split}
        &(1 - \frac{1}{2^p} (\frac{r+2}{r+1})^p)^{\lceil \log_2 r \rceil}
        >(1 - \frac{1}{2^p} (\frac{r+2}{r+1})^p)^{p+\log_2(p)}\\
        =& e^{(p+\log_2(p))\ln{(1 - \frac{1}{2^p} (\frac{r+2}{r+1})^p)}}.
    \end{split}
    \end{equation}
    According to $\ln x > \frac{1}{2}(x+\frac{1}{x})$ when $x\in (0,1)$, the power of $e$ follows
    \begin{equation}
    \begin{split}
        &(p+\log_2(p))\ln{(1 - \frac{1}{2^p} (\frac{r+2}{r+1})^p)}\\
        >& (p+\log_2(p)) \frac{1}{2}(-\frac{1}{d^p-1}-\frac{1}{d^p})\\
        >& -\frac{1}{2} (2p) (\frac{2}{d^p-1})=\frac{-2p}{d^p-1},
    \end{split}
    \end{equation}
    where $d=\frac{2(r+1)}{r+2}$.
\end{proof}

\section{Additional Analysis}
Because our above-mentioned algorithm is based on the result of the $QFT^\dagger$ is exact $(t-\tau-s)l \pmod{r}$, however, the $QFT^\dagger$ could lead to a superposition in many cases. Thus, in this section, we test whether the superposition will influence our upper conclusion. Moreover, we also propose an implementation of $\Gamma_\tau(M_a)$ which costs $O(nm^3)$ time.

\subsection{Algorithm}
With the state $|\phi_3\rangle$ before applying $QFT_{2^m}^\dagger$, we have
\begin{equation}
\begin{split}
    |\phi'_4\rangle =& (I^{\otimes n}\otimes QFT_{2^m}^\dagger \otimes I^{\otimes m} \otimes I)|\phi_3\rangle\\
    =&\frac{1}{2^{m}\sqrt{2^{n}r}} \sum^{r-1}_{l=0} \sum^{2^n-1}_{s=0} |s\rangle \left ( \sum^{2^m-1}_{\tilde{x}=0} \sum^{2^m - 1}_{x=0} e^{2 \pi i [(t -\tau - s)l/r - \tilde{x}/2^m] x} |\tilde{x}\rangle |\psi_l\rangle |0\rangle \right).
\end{split}
\end{equation}

Then, apply $U_g$ on the second and fourth registers, we get
\begin{equation}
\begin{split}
    |\phi'_5\rangle =\frac{1}{2^{m}\sqrt{2^{n}r}} \sum^{r-1}_{l=0} \sum^{2^n-1}_{s=0} |s\rangle & \left ( \sum^{2^m-1}_{\tilde{x}=1} \sum^{2^m - 1}_{x=0} e^{2 \pi i [(t -\tau - s)l/r - \tilde{x}/2^m] x} |\tilde{x}\rangle |\psi_l\rangle |0\rangle \right . \\
    & \left. + \sum^{2^m - 1}_{x=0} e^{2 \pi i [(t -\tau - s)l/r] x} |0\rangle |\psi_l\rangle |1\rangle \right ).
\end{split}
\end{equation}

To undo the operator applied before $|\phi'_5\rangle$, we apply their inverse operator to the quantum states $|\phi'_5\rangle$. Thus, after applying $QFT_{2^m}$ to the second register, the quantum state becomes
\begin{equation}
\begin{split}
    |\phi'_6\rangle =\frac{1}{2^{m}\sqrt{2^{m+n}r}} \sum^{r-1}_{l=0} \sum^{2^n-1}_{s=0} |s\rangle & \left ( \sum^{2^m-1}_{\tilde{x}=1} \sum^{2^m - 1}_{x=0} e^{2 \pi i [(t -\tau - s)l/r - \tilde{x}/2^m] x} \sum^{2^m-1}_{y=0} e^{2 \pi i \tilde{x} y /2^m}|y\rangle |\psi_l\rangle |0\rangle \right . \\
    & \left. + \sum^{2^m - 1}_{x=0} e^{2 \pi i [(t -\tau - s)l/r] x} \sum^{2^m-1}_{y=0} |y\rangle |\psi_l\rangle |1\rangle \right ).
\end{split}
\end{equation}
Then applying the inverse gate of $\Gamma_{\tau}(M_a)$ and $\Lambda(M_b)$, we have
\begin{equation}
\begin{split}
    |\phi'_8\rangle =\frac{1}{2^{m}\sqrt{2^{m+n}r}} \sum^{r-1}_{l=0} \sum^{2^n-1}_{s=0} |s\rangle & \left ( \sum^{2^m-1}_{\tilde{x}=1} \sum^{2^m - 1}_{x=0} e^{2 \pi i [(t -\tau - s)l/r - \tilde{x}/2^m] x} \sum^{2^m-1}_{y=0} e^{2 \pi i [\tilde{x}/2^m-(t -\tau - s)l/r] y }|y\rangle |\psi_l\rangle |0\rangle \right . \\
    & \left. + \sum^{2^m - 1}_{x=0} e^{2 \pi i (t -\tau - s)(x-y)l/r } \sum^{2^m-1}_{y=0} |y\rangle |\psi_l\rangle |1\rangle \right ).
\end{split}
\end{equation}
Applying $H$ gates to the first and second registers, the final state becomes
\begin{equation}
\begin{split}
    |\phi'_9\rangle =\frac{1}{2^{2m+n}\sqrt{r}} \sum^{r-1}_{l=0} \sum^{2^n-1}_{s=0} \sum^{2^n-1}_{s'=0} (-1)^{s \cdot s'} |s'\rangle & \left ( \sum^{2^m-1}_{\tilde{x}=1} \sum^{2^m - 1}_{x=0} e^{2 \pi i [(t -\tau - s)l/r - \tilde{x}/2^m] x} \sum^{2^m-1}_{y=0} e^{2 \pi i [\tilde{x}/2^m-(t -\tau - s)l/r] y }\right . \\ 
    &\times \sum^{2^m-1}_{y'=0} (-1)^{y \cdot y'}  
    |y'\rangle |\psi_l\rangle |0\rangle \\
    & \left. + \sum^{2^m - 1}_{x=0} e^{2 \pi i (t -\tau - s)(x-y)l/r } \sum^{2^m-1}_{y=0} \sum^{2^m-1}_{y'=0} (-1)^{y \cdot y'} |y'\rangle |\psi_l\rangle |1\rangle \right ).
\end{split}
\end{equation}

Then, only considering the state with the fourth register as $|1\rangle$, we also rewrite the $|\phi'_9\rangle$ as
\begin{equation}
\begin{split}
    \frac{1}{2^{2m+n}r} \sum^{2^n-1}_{s'=0} \sum^{2^m-1}_{y'=0} \sum^{r-1}_{k=0} \left[ \sum^{2^n-1}_{s=0} (-1)^{s \cdot s'} \sum^{2^m-1}_{y=0} (-1)^{y \cdot y'} \sum^{2^m-1}_{x=0} \sum^{r-1}_{l=0} e^{2 \pi i [(t-\tau-s) (x-y)-k] l/r} 
 \right] |s'\rangle |y'\rangle |a^k\rangle |1\rangle.
\end{split}
\end{equation}
Then if we have $t\in \{\tau+s \pmod{r}| s= 0,1,2,...,2^n-1 \}$, the quantum state whose third register is $|0^{m-1}1\rangle$ is
\begin{equation}
\begin{split}
    &\frac{1}{2^{2m+n}r} \sum^{2^n-1}_{s'=0} \sum^{2^m-1}_{y'=0} \left[ \sum^{2^n-1}_{s=0} (-1)^{s \cdot s'} \sum^{2^m-1}_{y=0} (-1)^{y \cdot y'} \sum^{2^m-1}_{x=0} \sum^{r-1}_{l=0} e^{2 \pi i [(t-\tau-s) (x-y)] l/r} \right] |s'\rangle |y'\rangle |0^{m-1}1\rangle |1\rangle \\
    =& \frac{1}{2^{2m+n}r} \sum^{2^n-1}_{s'=0} \sum^{2^m-1}_{y'=0} \left[ \sum^{2^n-1}_{s=0,s\neq t-\tau} (-1)^{s \cdot s'} \sum^{2^m-1}_{y=0} (-1)^{y \cdot y'} r
    + (-1)^{t\cdot s'} \sum^{2^m-1}_{y=0} (-1)^{y \cdot y'} 2^m r
    \right] |s'\rangle |y'\rangle |0^{m-1}1\rangle |1\rangle\\
    =& \frac{1}{2^{m+n}r} \sum^{2^n-1}_{s'=0} \left[ \sum^{2^n-1}_{s=0,s\neq t-\tau} (-1)^{s \cdot s'} r
    + (-1)^{t\cdot s'} 2^m r
    \right] |s'\rangle |0\rangle |0^{m-1}1\rangle |1\rangle.
\end{split}
\end{equation}

Otherwise, if the third register is not $|0^{m-1}1\rangle$, we have
\begin{equation}
\begin{split}
    &\frac{1}{2^{2m+n}r} \sum^{2^n-1}_{s'=0} \sum^{2^m-1}_{y'=0} \sum^{r-1}_{k=1} \left[ \sum^{2^n-1}_{s=0} (-1)^{s \cdot s'} \sum^{2^m-1}_{y=0} (-1)^{y \cdot y'} \sum^{2^m-1}_{x=0} \sum^{r-1}_{l=0} e^{2 \pi i [(t-\tau-s) (x-y)-k] l/r} \right] |s'\rangle |y'\rangle |a^k\rangle |1\rangle\\
    =&\frac{1}{2^{2m+n}r} \sum^{2^n-1}_{s'=0} \sum^{2^m-1}_{y'=0} \sum^{r-1}_{k=1} \left[ \sum^{2^n-1}_{s=0, s \neq t-\tau} (-1)^{s \cdot s'} \sum^{2^m-1}_{y=0} (-1)^{y \cdot y'} r \right] |s'\rangle |y'\rangle |a^k\rangle |1\rangle\\
    =&\frac{1}{2^{m+n}r} \sum^{2^n-1}_{s'=0} \sum^{r-1}_{k=1} \left[ \sum^{2^n-1}_{s=0, s \neq t-\tau} (-1)^{s \cdot s'} r \right] |s'\rangle |0\rangle |a^k\rangle |1\rangle.
\end{split}
\end{equation}

If $t \notin \{\tau+s \pmod{r}| s= 0,1,2,...,2^n-1 \}$, the state whose fourth register is $|1\rangle$ and third register is $|a^k\rangle$ for any $k$ is
\begin{equation}
\begin{split}
    &\frac{1}{2^{2m+n}r} \sum^{2^n-1}_{s'=0} \sum^{2^m-1}_{y'=0} \left[ \sum^{2^n-1}_{s=0} (-1)^{s \cdot s'} \sum^{2^m-1}_{y=0} (-1)^{y \cdot y'} \sum^{2^m-1}_{x=0} \sum^{r-1}_{l=0} e^{2 \pi i [(t-\tau-s) (x-y)-k] l/r} 
 \right] |s'\rangle |y'\rangle |a^k\rangle |1\rangle\\
    = &\frac{1}{2^{2m+n}r} \sum^{2^n-1}_{s'=0} \sum^{2^m-1}_{y'=0} \left[ \sum^{2^n-1}_{s=0} (-1)^{s \cdot s'} \sum^{2^m-1}_{y=0} (-1)^{y \cdot y'} (\sum^{2^m-1}_{x=0,(t-\tau-s) (x-y)-k \neq 0} \sum^{r-1}_{l=0} e^{2 \pi i [(t-\tau-s) (x-y)-k] l/r} 
 + \sum^{r-1}_{l=0} 1)\right]\\ &|s'\rangle |y'\rangle |a^k\rangle |1\rangle\\
 = &\frac{1}{2^{2m+n}r} \sum^{2^n-1}_{s'=0} \sum^{2^m-1}_{y'=0} \left[ \sum^{2^n-1}_{s=0} (-1)^{s \cdot s'} \sum^{2^m-1}_{y=0} (-1)^{y \cdot y'} r\right] |s'\rangle |y'\rangle |a^k\rangle |1\rangle\\
 = &\frac{1}{2^{m}} |0\rangle |0\rangle |a^k\rangle |1\rangle.
\end{split}
\end{equation}

The upper results represent that our algorithm can achieve similar results even if the output of $QFT_{2^m}^\dagger$ is a superposition state. Then, we analyze the success probability as follows.

\subsection{Success probability}

Following \cite{AA}, let $\tilde{x}$ be the the classical result of measuring the $QFT^\dagger$ of some quantum state $|S_{(t-\tau-s)l/r}\rangle=\frac{1}{\sqrt{r}}\sum^{r-1}_{j=0} e^{2\pi i (t-\tau-s)l j/r}|j\rangle$ in the computational basis. If $(t-\tau-s)l=0$, because $2^m (t-\tau-s)l/r$ is an integer, the probability of $\tilde{x}=0$ is $1$. For each $(t-\tau-s)l\neq 0$, the probability of $\tilde{x}=0$ follows
\begin{equation}
    Pr(\tilde{x}=0| (t-\tau-s)l\neq0) = \frac{\sin^2(N\Delta\pi)}{N^2\sin^2(\Delta\pi)} \leq \frac{1}{(2N\Delta)^2}
\end{equation}
where $\Delta=d((t-\tau-s)l/r, 0)$, and $d((t-\tau-s)l/r, 0) = \min_{z\in \mathbb{Z}}|z+(t-\tau-s)l/r|$. Therefore, we propose the following lemma to analyze the total probability of $\tilde{x}=0$, which represents the total probability:
\begin{lemma} \label{not_exact_4_1}
    Given integers $t,r,m,n \in \mathbb{Z}^+$, $s \in S$ such that $(t-\tau-s)l/r\in C$ where $C \subset \{x/r| x\in \mathbb{Z}, x/r\in [0,1)\}$, let $\tilde{x}$ be the discrete random variable corresponding to the classical result of measuring the second register of $|\phi'_5\rangle$ in the computational basis. If there is $s$ such that $(t-\tau-s)l/r=0$, then the probability of $\tilde{x}=0$ follows
    \begin{equation}
        Pr(\tilde{x}=0| \{0\} \subset C)< \frac{1}{2^m 2^n} + \frac{1}{2^n}+ \frac{1}{r}.
    \end{equation}
    Else if there is not $s$ such that $(t-\tau-s)l/r=0$, then the probability of $\tilde{x}=0$ follows
    \begin{equation}
        Pr(\tilde{x}=0| \{0\} \not\subset C)> \frac{1}{r}.
    \end{equation}
\end{lemma}
\begin{proof}
    Considering $|\phi'_5\rangle$, the probability of the fourth register is $1$ equals
    \begin{equation}
    \begin{split}
        \frac{1}{2^{2m}2^{n}r}& \left |\sum^{r-1}_{l=0} \sum_{s\in S} \langle s| \sum^{2^m - 1}_{x=0} e^{-2 \pi i [(t -\tau - s)l/r] x} \langle 0| \langle \psi_l| \langle 1| \right .\\
        &\left . \sum^{r-1}_{l=0} \sum_{s\in S} |s\rangle \sum^{2^m - 1}_{x=0} e^{2 \pi i [(t -\tau - s)l/r] x} |0\rangle |\psi_l\rangle |1\rangle \right | \\
        = & \frac{1}{2^{2m}2^{n}r} \sum^{r-1}_{l=0} \sum_{s\in S} \left (\sum^{2^m - 1}_{x=0} e^{-2 \pi i [(t -\tau - s)l/r] x} \right ) \left (\sum^{2^m - 1}_{x=0} e^{2 \pi i [(t -\tau - s)l/r] x} \right )\\
        < & \frac{1}{2^{2m}2^{n}r} \sum^{r-1}_{l=1}\sum_{s\in S, s\neq t-\tau} \frac{1}{(2\Delta)^2}  + \frac{1}{r}+ \frac{1}{2^n}\\
        < & \frac{1}{2^{2m}2^{n}r} \sum^{r-1}_{l=1}\sum_{s\in S, s\neq t-\tau} \frac{2 r^2}{(2(t-\tau-s)l)^2}  + \frac{1}{r}+ \frac{1}{2^n}\\
        < & \frac{r^2}{2^{2m+n+1}r} \sum^{r-1}_{l=1}\frac{1}{l^2}\sum_{s\in S, s\neq t-\tau} \frac{1}{(t-\tau-s)^2}  + \frac{1}{r}+ \frac{1}{2^n}\\
        < & \frac{r^2}{2^{2m+n+1}r} \frac{\pi^4}{36}  + \frac{1}{r}+ \frac{1}{2^n} < \frac{1}{2^{m+n}} + \frac{1}{r}+ \frac{1}{2^n}.
    \end{split}
    \end{equation}
    Similarly, we can compute the second half with a similar process:
    \begin{equation}
    \begin{split}
        \frac{1}{2^{2m}2^{n}r}& \left|\sum^{r-1}_{l=0} \sum_{s\in S} \langle s| \sum^{2^m - 1}_{x=0} e^{-2 \pi i [(t -\tau - s)l/r] x} \langle 0| \langle \psi_l| \langle 1| \right. \\
        &\left. \sum^{r-1}_{l=0} \sum_{s\in S} |s\rangle \sum^{2^m - 1}_{x=0} e^{2 \pi i [(t -\tau - s)l/r] x} |0\rangle |\psi_l\rangle |1\rangle \right| \\
        = & \frac{1}{2^{2m}2^{n}r} \sum^{r-1}_{l=0} \sum_{s\in S} \left (\sum^{2^m - 1}_{x=0} e^{-2 \pi i [(t -\tau - s)l/r] x} \right ) \left (\sum^{2^m - 1}_{x=0} e^{2 \pi i [(t -\tau - s)l/r] x} \right )\\
        = & \frac{1}{2^{2m}2^{n}r} \sum^{r-1}_{l=1} \sum_{s\in S} \left (\sum^{2^m - 1}_{x=0} e^{-2 \pi i [(t -\tau - s)l/r] x} \right ) \left (\sum^{2^m - 1}_{x=0} e^{2 \pi i [(t -\tau - s)l/r] x} \right ) + \frac{1}{r}
        > \frac{1}{r}.
    \end{split}
    \end{equation}
\end{proof}
According to Proposition \ref{t_in} and \ref{t_notin}, if the $QFT^\dagger$ is exact $(t-\tau-s)l \pmod{r}$, the probability of observing $0$ is no less than $\frac{1}{2^n}$ and $\frac{1}{r}$ separately. Then, knowing that Lemma \ref{not_exact_4_1}, we have
\begin{lemma} \label{not_exact_3_1}
    In our algorithm, for integer $t,r,m,n,\tau \in \mathbb{Z}^+$, if $t\in \{\tau+s \pmod{r}| s= 0,1,2,...,2^n-1 \}$, then the probability of the fourth register is $1$ and the third register is $1$ is no less than $\frac{(2^m-1)^2}{2^{2m+n}}$. And otherwise if the fourth register is $1$ and $t\notin \{\tau+s \pmod{r}| s= 0,1,2,...,2^n-1 \}$, the probability of observing $0^{m-1}1$ from the third and fourth registers is $\frac{1}{2^{2m}}$.
\end{lemma}
\begin{proof}
    With quantum state $|\phi'_9\rangle$, if $t\in \{\tau+s \pmod{r}| s= 0,1,2,...,2^n-1 \}$, the probability of the third register is $0^{m-1}1$ and the fourth register is $1$ equals
    \begin{equation}
    \begin{split}
    &\left |\frac{1}{2^{m+n}r} \sum^{2^n-1}_{s'=0} \left[ \sum^{2^n-1}_{s=0,s\neq t-\tau} (-1)^{s \cdot s'} r
    + (-1)^{t\cdot s'} 2^m r
    \right] \langle s'| \langle0| \langle 0^{m-1}1| \langle1| \right .\\
    &\left . \frac{1}{2^{m+n}r} \sum^{2^n-1}_{s'=0} \left[ \sum^{2^n-1}_{s=0,s\neq t-\tau} (-1)^{s \cdot s'} r
    + (-1)^{t\cdot s'} 2^m r
    \right] |s'\rangle |0\rangle |0^{m-1}1\rangle |1\rangle \right |\\
    = & \frac{1}{2^{2m+2n}r^2} \sum^{2^n-1}_{s'=0} \left[ \sum^{2^n-1}_{s=0,s\neq t-\tau} (-1)^{s \cdot s'} r
    + (-1)^{t\cdot s'} 2^m r
    \right] \left[ \sum^{2^n-1}_{s=0,s\neq t-\tau} (-1)^{s \cdot s'} r
    + (-1)^{t\cdot s'} 2^m r \right] \\
    = & \frac{1}{2^{2m+2n}r^2} \sum^{2^n-1}_{s'=0} \left[ -(-1)^{t \cdot s'} r
    + (-1)^{t\cdot s'} 2^m r
    \right] \left[ -(-1)^{t \cdot s'} r
    + (-1)^{t\cdot s'} 2^m r \right]\\
    = &\frac{1}{2^{2m+2n}r^2} (\sum^{2^n-1}_{s'=1}(2^m-1)^2 r^2 + ((2^n-1)r+2^m r )^2)\\
    > &\frac{(2^n-1)(2^m-1)^2 r^2 + (2^m+2^n-1)^2r^2}{2^{2m+2n}r^2}\\
    > &\frac{(2^n-1)(2^m-1)^2 + (2^{m}-1)^2}{2^{2m+2n}}= \frac{2^n(2^{m}-1)^2}{2^{2m+2n}} = \frac{(2^m-1)^2}{2^{2m+n}}.
    \end{split}
    \end{equation}
    If $t\notin \{\tau+s \pmod{r}| s= 0,1,2,...,2^n-1 \}$, the probability of observing $1$ from the third register is
    \begin{equation}
    \begin{split}
        &\left |\frac{1}{2^{m}} \langle0| \langle0| \langle a^k| \langle1| \frac{1}{2^{m}} |0\rangle |0\rangle |a^k\rangle |1\rangle \right| = \frac{1}{2^{2m}}.
    \end{split}
    \end{equation}
    
\end{proof}
Combine the Lemma \ref{not_exact_4_1} with \ref{not_exact_3_1}, we can compute the probability of observing $0^{m-1}1$ from the third register if the fourth register is $1$:
\begin{lemma}
    In our algorithm, given integers $t,r,m,n,\tau \in \mathbb{Z}^+$, if $t\in \{\tau+s \pmod{r}| s= 0,1,2,...,2^n-1 \}$ and the fourth register is $1$, then the probability of obtaining $0^{m-1}1$ from the third register is no less than:
    \begin{equation}
        \frac{\frac{(2^m-1)^2}{2^{2m+n}}}{\frac{1}{2^m 2^n} + \frac{1}{2^n}+ \frac{1}{r}}>\frac{(2^m-1)^2}{2^{2m+1}}.
    \end{equation}
    Otherwise, if $t\notin \{\tau+s \pmod{r}| s= 0,1,2,...,2^n-1 \}$ and the fourth register is $1$, the probability of obtaining $0^{m-1}1$ from the third register is no greater than:
    \begin{equation}
        \frac{\frac{1}{2^{2m}}}{\frac{1}{r}} = \frac{r}{2^{2m}} <\frac{1}{2^m}.
    \end{equation}
\end{lemma}
\begin{proof}
As proved in Lemma \ref{not_exact_4_1} and \ref{not_exact_3_1}, if $t\in \{\tau+s \pmod{r}| s= 0,1,2,...,2^n-1 \}$ and the fourth register is $1$, the probability of obtaining $0^{m-1}1$ from the third register is no less than
\begin{equation}
\begin{split}
    \frac{\frac{(2^m-1)^2}{2^{2m+n}}}{\frac{1}{2^m 2^n} + \frac{1}{2^n}+ \frac{1}{r}}
    = \frac{(2^m-1)^2}{2^m + 2^{2m} + \frac{2^{2m+n}}{r}} > \frac{(2^m-1)^2}{2^{2m} + 2^{2m}} = \frac{(2^m-1)^2}{2^{2m+1}}.
\end{split}
\end{equation}
Then, if $t\in \{\tau+s \pmod{r}| s= 0,1,2,...,2^n-1 \}$ and the fourth register is $1$, we also have the similar results.
\end{proof}
According to this lemma, the probability of observing the $1$ from the third register is still very distinguishable even if the $QFT$ is not exact. Thus, as discussed in Section \ref{Strategy}, if we measure the output of our algorithm for $p$ times, the success probability of our algorithm can be described as:
\begin{theorem}
    With our Algorithm \ref{algorithm_strategy}, if the number of itereations is $p$, then the success probability $\mathcal{P}'$ follows
    \begin{equation}
        \mathcal{P}'>\left(1-(1-\frac{(2^m-1)^2}{2^{2m+1}})^{p}\right)^{\lceil \log_2 r \rceil} > \frac{1}{e^{\frac{2p}{d^p-1}}} ,
    \end{equation}
    where $p+\log_2 p \leq \log_2 r$ and $d=\frac{2(r+1)}{r+2}$.
\end{theorem}
\begin{proof}
    Knowing that we need to observe $0^{m-1}1$ for at least $1$ time. The probability of obtaining no $1$ in algorithm \ref{DQDLA} is 
    \begin{equation}
        (1-\frac{(2^m-1)^2}{2^{2m+1}})^{p}.
    \end{equation}
    Thus we have the probability of observing $0^{m-1}1$ at least $1$ time is 
    \begin{equation}
    \begin{split}
        &(1-(1-\frac{(2^m-1)^2}{2^{2m+1}})^{p})^{\lceil \log_2 r \rceil} \\
        =& (1-(\frac{2^{2m+1}-(2^m-1)^2}{2^{2m+1}})^p)^{\lceil \log_2 r \rceil}\\
        =& (1-(\frac{2^{2m}+2^{m+1}-1}{2^{2m+1}})^p)^{\lceil \log_2 r \rceil}\\
        =& (1-(\frac{1}{2}\frac{2^{2m}+2^{m+1}-1}{2^{2m}})^p)^{\lceil \log_2 r \rceil}.
    \end{split}
    \end{equation}
    Compared with the result in Theorem \ref{success_prob}, because $r+1 \leq \frac{2^m}{2}$ we have
    \begin{equation}
    \begin{split}
        \frac{r+2}{r+1} =& 1 + \frac{1}{r+1}
        \geq 1 + \frac{2}{2^m} \\
        > & 1 + \frac{2}{2^m} - \frac{1}{2^{2m}} = \frac{2^{2m}+2^{m+1}-1}{2^{2m}}.
    \end{split}
    \end{equation}
    Then according to the results of Theorem \ref{success_prob}, we have 
    \begin{equation}
        \mathcal{P}' > \left(1-(\frac{1}{2}\frac{2^{2m}+2^{m+1}-1}{2^{2m}})^{p}\right)^{\lceil \log_2 r \rceil} > \left(1 - \frac{1}{2^p} (\frac{r+2}{r+1})^p \right )^{\lceil \log_2 r \rceil}>\frac{1}{e^{\frac{2p}{d^p-1}}}.
    \end{equation}
\end{proof}
This theorem gives a lower bound of the success probability of our algorithm. For example, if given $r=2^{12}$, then we can set $p=8$, and the success probability $\mathcal{P}'$ will be greater than $89.24\%$, which is already a sufficiently high probability of success. 

\subsection{The implementation of $\Gamma_{\tau}(M_{a})$} \label{subsec:impl}
In this work, we introduce a gate $\Gamma_{\tau}(M_{a})$ with two control registers. Then, we will discuss how to implement it. 

As introduced by Shor \cite{DLP_first}, for integer $a$, we can design a gate such that for any integer $0 \leq z < N$,
\begin{equation}
    M_a|z\rangle = |za \pmod{N}\rangle.
\end{equation}
And 
\begin{equation}
    \Lambda(M_a)|x\rangle|z\rangle = |x\rangle (M_a)^x|z\rangle.
\end{equation}

For any $U: \{0,1\}^m \times \{z\in \mathbb{Z}|0 \leq z < N\} \to \{0,1\}^m \times \{z\in \mathbb{Z}|0 \leq z < N\}$, we further denote a quantum gate
\begin{equation}
    \Xi(U):|s\rangle|x\rangle|z\rangle \to |s\rangle U^s(|x\rangle|z\rangle).
\end{equation}
Therefore, we have
\begin{equation}
\begin{split}
    \Xi(\Lambda(M^\dagger_a))(|s\rangle|x\rangle|z\rangle) =& |s\rangle (\Lambda(M^\dagger_a))^s(|x\rangle|z\rangle)\\
    =&|s\rangle|x\rangle(M^\dagger_a)^{sx}|z\rangle\\
    =&|s\rangle|x\rangle|za^{-sx} \pmod{N}\rangle.
\end{split}
\end{equation}
With this gate, we can construct our $\Gamma_{\tau}(M^\dagger_{a})=(I^{\otimes n} \otimes \Lambda(M^\dagger_{a^\tau}) \Xi(\Lambda(M^\dagger_a))$:
\begin{equation}
\begin{split}
    \Gamma_{\tau}(M^\dagger_{a})(|s\rangle|x\rangle|z\rangle)
    =& (I^{\otimes n} \otimes \Lambda(M^\dagger_{a^\tau}) \Xi(\Lambda(M^\dagger_a)) (|s\rangle|x\rangle|z\rangle)\\
    =& (I^{\otimes n} \otimes \Lambda(M^\dagger_{a^\tau}) (|s\rangle|x\rangle|za^{-sx} \pmod{N}\rangle)\\
    =&|s\rangle|x\rangle(M^\dagger_{a^\tau})^x|za^{-sx} \pmod{N}\rangle\\
    =&|s\rangle|x\rangle|za^{-(s+\tau)x} \pmod{N}\rangle,
\end{split}
\end{equation}
which can be applied in $O(nm^3)$ time.

\section{Complexity}
After analyzing the success probability of our algorithm, we introduce the time, space, and communication complexity. The results are summarized in Table \ref{tab:compare}, which represents that our algorithm can reduce the Qubit number and enhance the success probability without significantly increasing the circuit depth compared with Shor's algorithm \cite{DLP_first}. Compared with the semi-classical QFT-based methods \cite{DLP_SCQFT}, our algorithm has shallower circuit depth and requires no mid-circuit measurement, which further deepens circuit depth and brings additional operations. Furthermore, our algorithm can work without heuristics or assumptions except the basic setting of DLPs.

\begin{table}[t]
    \centering
    \small
    \setlength{\tabcolsep}{4pt}
    \renewcommand{\arraystretch}{1.15}
    \caption{Comparison between our algorithm and other existing methods.}
    \label{tab:compare}
    \begin{tabular}{p{2.6cm}|c c p{3.2cm} p{2.2cm}}
        \toprule
        \centering Algorithm
        & \centering Circuit depth
        & \centering Qubit number
        & \centering Success probability $(r\to\infty)$
        & \centering Assumption \tabularnewline
        \midrule
        \centering Shor's \cite{DLP_first}
        & $O(m^3)$
        & $O(3m)$
        & \centering $>0.657$
        & \centering Basic \tabularnewline

        \centering Semi-classical QFT methods \cite{DLP_SCQFT}
        & $O(m^4)$
        & $O(2m)$
        & \centering unknown
        & \centering With mid-circuit measure \tabularnewline

        \centering Chevignard et al. \cite{DLP_reducing}
        & $O(m^2 \log^3(m))$
        & $O(\frac{7m}{2})$
        & \centering $>0.9$
        & \centering Heuristic \tabularnewline

        \centering Ours
        & $O(nm^3)$
        & $O(2m+n+1)$
        & \centering $>\exp\!\left[-\frac{2p}{2^p-1}\right]$
        & \centering Basic \tabularnewline
        \bottomrule
    \end{tabular}
\end{table}

\subsection{Time and space complexity}
In quantum computing, the key influencing factors include circuit depth, the number of qubits, and the success probability. The circuit depth and the number of qubits are the fundamental limits of today's quantum computers \cite{NISQ}, while the success probability shows the reliability of our algorithm.

The following theorem describes the time and space complexity and circuit depth we need in Algorithm \ref{algorithm_strategy}:

\begin{theorem}
    Given $a,b,r,N \in \mathrm{Z}^+$, $a^r=1 \pmod{N}$ and promise exist $t$ such that $a^t = b \pmod{N}$. Let $m=\lceil \log_2 r \rceil+1$, $n<m-1$, then the quantum process (Algorithm \ref{DQDLA}) requires $O(n m^3)$ circuit depth and at most $O(2m+n+1)$ qubits, while our Algorithm \ref{algorithm_strategy} cost total $O(m^5n2^n/K)$ time to find the promised $t$ with $K$ QPU.
\end{theorem}
\begin{proof}
    The time complexity of our algorithm can be computed as the combination of each gate. Knowing that each Hadamard gate costs $O(1)$, the $\Lambda(M_b)$ gates can be applied in $O(m^3)$ times, and the $QFT$ takes $O(m^2)$. Then, our $\Gamma_{\tau}(M_a)$ can be implemented via method shown in section \ref{subsec:impl}. Following our implementation, $\Xi(\Lambda(M^\dagger_a))$ costs $O(nm^3)$ time because $\Lambda(M^\dagger_a)$ takes $O(m^3)$ and the first register contains $n$ qubits and $U_{g}$ only take constant time. Thus, the quantum part of our algorithm (Algorithm \ref{DQDLA}) can query whether $t$ is contained in each set $\tau$ in $O(2m^3+2nm^3+2m^2+2m+2n+2)$ time, which is $O(nm^3)$.   
    
    Knowing that the probability of observing $1$ from the fourth register follows $O(\frac{1}{2^n})$, we need $O(2^n)$ averagely. To locate a specific $t$, our algorithm needs to query $O(\log_2 r)$ sets at most. In our paper, the number of queries has the same order with $O(m)$ and the iteration time follows $p=O(\log_2 r)$. Although the algorithm requires multiple iterations, the time for each iteration is short. Therefore, the circuit depth will not be significantly increased as a result. 
    
    Thus, the total time complexity is $O(m^5n2^n)$ (short for $O(m^3  (\log_2 r)^2 n 2^n)$) if our algorithm runs serially. Moreover, if we have $K$ QPUs, our time complexity can be reduced to $O(m^5n2^n/K)$ because the quantum part can run in parallel, e.g., the quantum algorithms for each set $S_{n,\tau}$. 

    For space complexity, our four registers contain $n,m,m,1$ qubits separately. Thus, our algorithm costs $O(2m+n+1)$ space. Once we limit the size of the set $\tau$ by letting $n-1<m$, our algorithm has lower space complexity than Shor's algorithm (space complexity $O(3m)$).
\end{proof}

Additionally, our Algorithm \ref{DQDLA} can be accelerated via Quantum Amplitude Amplification Algorithm \cite{AA}. In this case, the Algorithm \ref{DQDLA} for different $S_{n,\tau}$ can still run parallel. Therefore, we have the total time complexity as $O(m^5n\sqrt{2^n}/K)$ when we split each set into  $K$ subsets which contains $2^n$ elements. However, applying Quantum Amplitude Amplification Algorithm leads to a deeper quantum circuit. Therefore, this scheme can be regarded as a trade off to the original one.

\subsection{Communication complexity}
Our distributed algorithm requires no quantum communication. In this section, we are discussing the classical communication required in our algorithm. If we run our algorithm serially, the latter QPU should know two lines of information from the former: First, the QPU should know which sets have been queried; Second, whether the solution $t$ is contained in the former queried set. Knowing that each of the sets can be described with two integers $\tau$ and $n$, we apply a $1$-bit integer to represent whether the solution is contained in the set. $\tau$ has the same scale as $r$ and $n$ has the same scale with $\log_2 r$. Therefore, the algorithm requires $O(\log_2 ((r+\log_2 r) \log_2 r))$ bits classical communications. It can also be written as $O(\log_2 (r \log_2 r))$ for short.

\section{Numerical results}
In this section, we are applying our algorithm onto a Small-scale discrete logarithm problem and presenting the numerical results. 

\begin{figure}
    \centering
    \includegraphics[width=0.8\linewidth]{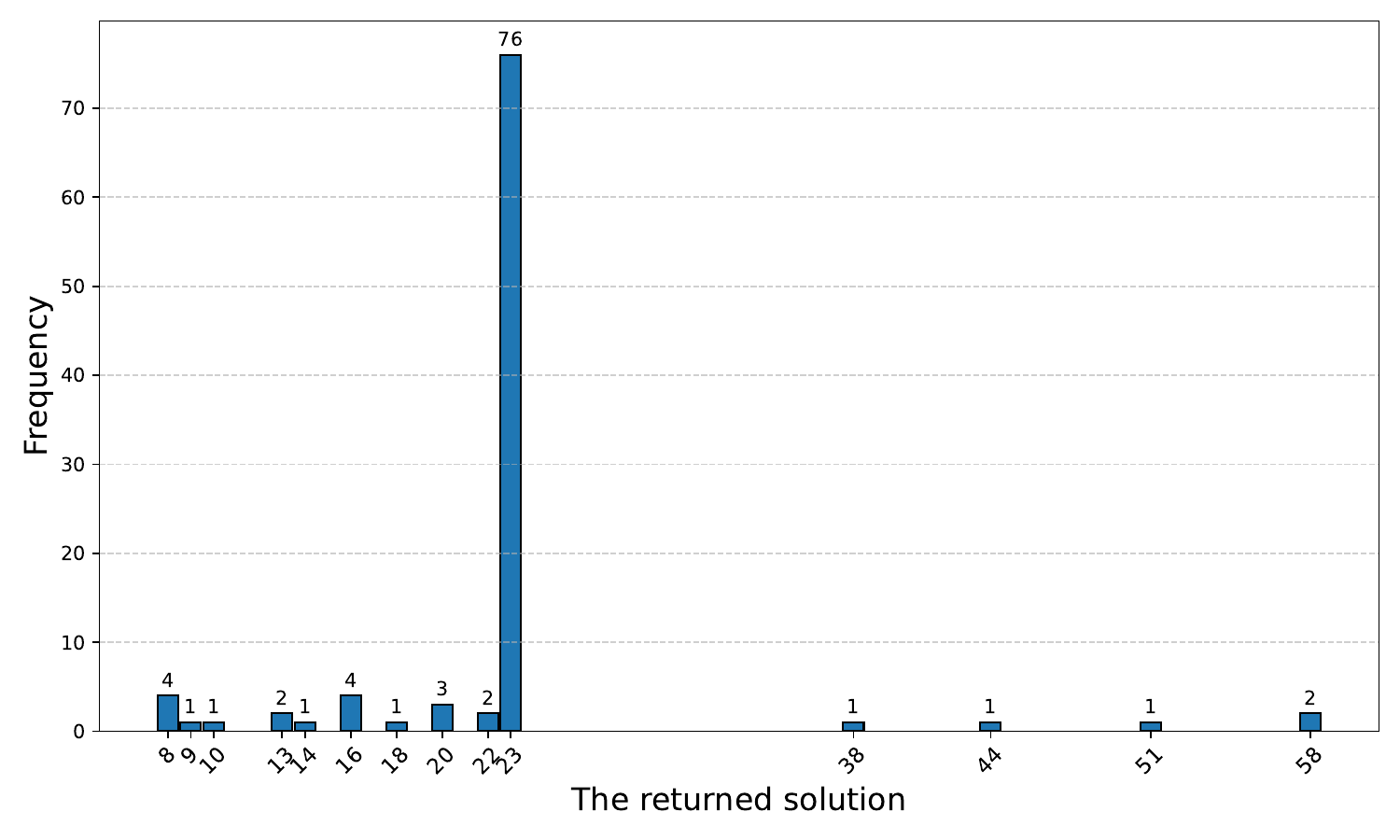}
    \caption{The frequency of the solutions ($t=23$) found by Algorithm \ref{algorithm_strategy} for $100$ times.}
    \label{fig:Barplot_Frequency}
\end{figure}

\begin{figure}
    \centering
    \includegraphics[width=1\linewidth]{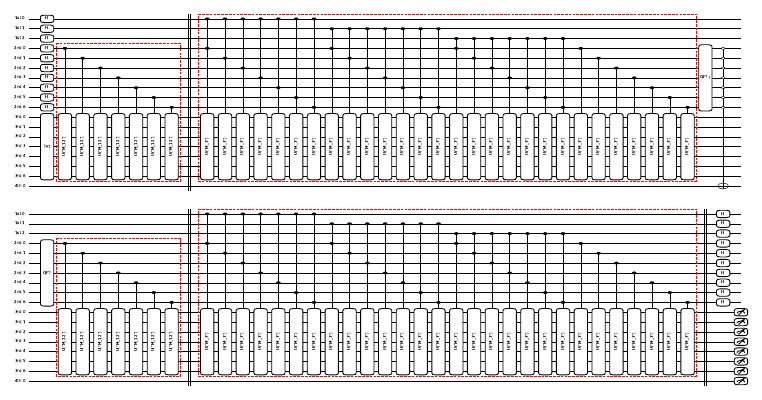}
    \caption{The circuit constructed via Pennylane.}
    \label{fig:Circuit_total}
\end{figure}

The discrete logarithm problem is defined with the following numbers: Given $a=3, b=12, N=71$, and therefore $r=35$, the purpose of our Algorithm is to find $t$ such that $3^t = 12 \pmod{71}$ ($t=23$ as a result). According to the previous analysis, we set $m=7$ (with $\epsilon=0.5$), $n=3<m-1$, $p=2$. Therefore, after constructing $M_a$ and $M_b$ and other gates via Pennylane \cite{pennylane}, a sample of circuits we construct is shown in Figure \ref{fig:Circuit_total}.

We run our Algorithm \ref{algorithm_strategy} for $100$ times and the frequency of the solutions is shown in Figure \ref{fig:Barplot_Frequency}. From Figure \ref{fig:Barplot_Frequency}, the frequency of our algorithm for finding the correct answer $t=23$ is $76$. Thus, our algorithm can find the answer for DLP in most cases. Under such setting of small scale problem, the theoretical lower success probability bound is only $\mathcal{P}>0.2380$, which introduces a significant gap between the theoretical bound and numerical result. The reason why this gap exists is that our analysis ignores some special cases. For example, when $n=3$ and $\tau=0$, our Algorithm \ref{DQDLA} accidentally returns True and then we set $n=2$ and $\tau=0$. However, our Algorithm \ref{DQDLA} returns all False until returning True at $n=2$ and $\tau=20$ and puts the Algorithm \ref{algorithm_strategy} back on track again. Those effects can lead to additional correct solutions beyond theoretical analysis.

\begin{figure}
    \centering
    \includegraphics[width=1\linewidth]{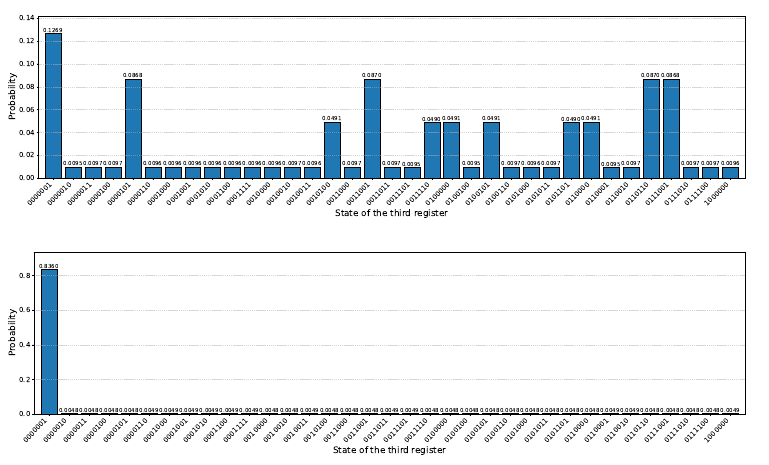}
    \caption{The probability of observing different strings from the third register when $t \not\in S_{n,r}$ (upper figure) and $t\in S_{n,r}$ (lower figure).}
    \label{fig:Barplot_prob}
\end{figure}

Then, we compare the probability of observing $0^{m-1}1$ at $n=3$ and $\tau=0$ or $20$ from the third register in Figure \ref{fig:Barplot_prob}. If $t \not\in S_{n,r}$ ($\tau=0$), the probability of observing $0^{m-1}1$ is $0.1269$. It is worth mentioning that this probability will continue to decrease if the scale of the problem increases. In the meantime, if $t \in S_{n,r}$ ($\tau=20$), the probability of observing $0^{m-1}1$ is $0.8360$, which is significant large comparing with $0.1269$. Therefore, our algorithm can easily find the appropriate $S_{n,r}$ that contains $t$.

\section{Conclusion}
Quantum discrete logarithm algorithms have received significant attention in recent years. However, the success probability and problem size are insufficient. In this paper, we proposed a distributed quantum algorithm for the discrete logarithm problem. Compared with the original algorithm, our method requires fewer qubits and can provide a higher success probability. The main idea is to exploit the fact that, when the fourth register is measured with outcome 1, the quantum state of the third register becomes distinguishable depending on whether the solution is contained in a given set. Because the set size can be flexibly adjusted, the algorithm can be distributed across multiple independent computational nodes. This feature makes the proposed framework both scalable and adaptable. At the same time, the current method still has limitations, since the probability of obtaining the desired outcome in the fourth register may remain low. Future work will focus on improving this probability by designing better mechanisms for preparing distinguishable superposition states, as well as extending the same idea to other computational problems.

%%
%% Bibliography
%%

%% Please use bibtex, 
\bibliographystyle{ieeetr}
\bibliography{ref}

\end{document}